\documentclass[12pt]{article}

\usepackage[a4paper,margin=1in]{geometry}
\usepackage[T1]{fontenc}
\usepackage{microtype}

\usepackage{mathtools}
\usepackage{amsthm}

\usepackage{wasysym}
\usepackage{amssymb}

\usepackage{algorithm}
\usepackage{algpseudocode}

\usepackage{tikz}
\usetikzlibrary{positioning,fit,backgrounds,patterns}

\usepackage{xcolor}
\usepackage{booktabs}
\usepackage{colortbl}
\usepackage{array}
\usepackage{tabularx}
\usepackage{makecell}

\usepackage[hidelinks]{hyperref}
\usepackage[nameinlink,noabbrev]{cleveref}

\newtheorem{theorem}{Theorem}[section]
\newtheorem{lemma}[theorem]{Lemma}
\newtheorem{proposition}[theorem]{Proposition}
\newtheorem{corollary}[theorem]{Corollary}

\theoremstyle{definition}
\newtheorem{definition}[theorem]{Definition}
\newtheorem{example}[theorem]{Example}
\newtheorem{problemstatement}[theorem]{Problem}
\newtheorem{openquestion}[theorem]{Open Question}

\theoremstyle{remark}
\newtheorem{remark}[theorem]{Remark}

\title{%
  Right Divisibility in Erasing Semi-Thue Systems:\\
  A Minimal View of Intruder Deduction
}
\usepackage{authblk}
\author[1]{Raja O. P. Damanik} \author[1]{Alwen Tiu} 
\affil[1]{
    School of Computing, Australian National University\\ Canberra, Australia\\ \texttt{\{raja.damanik,alwen.tiu\}@anu.edu.au} 
}

\date{}

\begin{document}

\maketitle

\begin{abstract}
The intruder deduction problem is central to symbolic security-protocol analysis: it asks whether an attacker can derive a target message from observed messages using (cryptographic) operators available to the attacker. Although convergent rewrite systems provide canonical normal forms, deduction modulo convergent theories remains undecidable in general, and existing decidable fragments are often shaped by practical cryptographic examples.
In this paper, we study deduction from a minimal structural perspective. When all function symbols are unary, terms collapse to words and deduction becomes a right-divisibility problem for semi-Thue systems: given words $u$ and $v$ decide whether there exists $w$ such that $wu \equiv_S v$. We investigate this problem for several classes of semi-Thue systems and prove, to the best of our knowledge, new decidability results for convergent prefix-erasing and convergent suffix-erasing systems.
We then extend this perspective to term rewriting systems whose rules erase contexts while lifting selected subterms or variables. Although these classes suggest possible decidable generalisations beyond the unary setting, we show that deduction is already undecidable for a convergent simultaneous variable-lifting system. This exposes both the potential and the limits of extending the right-divisibility results to richer equational theories.
\end{abstract}

\section{Introduction}

Deduction is central to the formal analysis of cryptographic protocols. In the symbolic, or Dolev--Yao, model \cite{dolev2003security}, messages are represented as abstract terms built from atomic names and public operations available to an attacker, such as pairing, encryption, hashing, and modular exponentiation. Given a collection of observed messages, the intruder deduction problem~\cite{Comon-LundhS03} asks whether the attacker can construct a target term, such as a secret key, nonce, or plaintext. It therefore captures a fundamental reachability question underlying many symbolic verification methods.

Deduction modulo an equational theory is undecidable in general, even for some theories admitting convergent presentations \cite{cortier2007deciding}. Convergent rewrite systems are nevertheless attractive because they provide canonical representatives: two terms are equivalent exactly when their normal forms coincide. Since convergence alone does not ensure decidability of deduction, an important objective is to identify fragments for which deduction is decidable.

This problem is particularly relevant because modern cryptographic protocols employ algebraic operations from many areas of mathematics \cite{cortier2006survey}, including abelian groups, exponentiation, etc. Several decidable classes have therefore been developed, including subterm-convergent theories \cite{cortier2007deciding}, permutative theories \cite{erbatur2024deciding}, contracting theories \cite{bunch2024knowledge}, and related fragments. These results also support various automated protocol-analysis tools \cite{barbosa2021sok}, such as ProVerif \cite{blanchet2016modeling}, DEEPSEC \cite{cheval2018deepsec}, KISS \cite{ciobacua2012computing}, YAPA \cite{baudet2013yapa}, and SPEC \cite{tiu2016spec}.

Most existing decidable classes are motivated by concrete cryptographic primitives. Although this practical focus is essential, it can produce technically intricate conditions whose structural significance is difficult to isolate. It is therefore not always clear which restrictions are intrinsically necessary for decidability or how far their underlying ideas can be generalised.

This paper takes a structurally minimal approach. We restrict attention to unary function symbols, where terms correspond to words, equational theories to semi-Thue systems~\cite{book1993string}, and deduction to a right-divisibility problem. We study the decidability of right-divisibility for several classes of semi-Thue systems, connecting classical string-rewriting problems with symbolic protocol verification. Positive results in this setting suggest principles that may extend to richer theories, while negative results identify boundaries that also apply to more general systems. In this way, minimal deduction problems provide a structural guide to the decidability landscape of symbolic cryptographic reasoning.

{\em \bf Contributions.} 
Motivated by deduction in symbolic cryptographic analysis, this paper studies the right-divisibility problem for semi-Thue systems: given $S=(\Gamma,R)$ and words $u,v\in\Gamma^*$, determine whether there exists $w\in\Gamma^*$ such that $wu\equiv_S v$. When all function symbols are unary, terms can be represented as words, and deduction reduces naturally to this problem. We prove that right-divisibility is decidable for convergent prefix-erasing systems and for convergent suffix-erasing systems, also known as dwindling systems~\cite{Akcam2021}. To the best of our knowledge, these results have not previously been isolated in this form.

These results are part of a broader programme connecting classical problems in string rewriting with the decidability of knowledge problems in symbolic protocol verification. We introduce related classes of term rewriting systems whose rules erase contexts while lifting selected subterms or variables. We conjecture decidability for the ordinary convergent lifting variants, but prove that deduction is undecidable for a convergent simultaneous variable-lifting system. Together, these results illustrate how minimal rewriting structures can help locate the boundary between decidable and undecidable deduction.

{\em \bf Related Works.}
The word problem and the left- and right-divisibility problems are classical decision problems for monoids presented by semi-Thue systems. They have been studied alongside Green's relations, finiteness, the group property, and related algebraic questions
\cite{book1983decidable,book1993string,otto1984some,otto1986two,
otto1991decision,otto1985deciding,otto1997properties,otto1998equational}.
A key objective is to understand how structural restrictions on a semi-Thue system affect the decidability and complexity of these problems.
Table~\ref{tab:word-rdiv-comparison} summarises representative results for the word and right-divisibility problems. 
For convergent systems, the word problem is decidable by comparing normal forms, but right divisibility remains undecidable \cite{book1982confluent}. A significant positive boundary is given by finite confluent monadic systems, for which right divisibility is decidable \cite{book1983decidable}.

This paper seeks decidability results beyond the monadic boundary. Rather than requiring every right-hand side to have length at most one, we study semi-Thue systems with structured erasing rules whose right-hand sides may be longer. This is motivated by subterm-convergent theories, for which deduction is decidable \cite{abadi2006deciding}: over unary signatures, such theories naturally induce prefix-erasing string-rewriting systems. Our approach is also related to earlier work on structured string rewriting for equational matching and unification \cite{otto1998equational}. Semi-Thue systems thus provide a minimal setting in which to isolate the structural sources of decidability.

In the more general setting of term rewriting, deduction is commonly studied together with static equivalence as knowledge problems. Subterm-convergent theories form a prominent decidable class, whereas richer algebraic theories are needed to model operations such as exclusive-or and other cryptographic primitives \cite{cortier2006survey}. Even for convergent theories, however, deduction becomes undecidable once the permitted rewriting structure is sufficiently general \cite{damanik2026term}.

One major line of research therefore seeks broad syntactic classes for which deduction, and often static equivalence, remain decidable. Permutative theories extend the setting by allowing controlled permutations of variables while preserving suitable structural invariants \cite{erbatur2024deciding}. In another direction, graph and and homeomorphic-embedding theories extend subterm-convergent systems by formalizing structural contraction through embedding relations \cite{bunch2024knowledge}. The related notion of contracting term-rewriting systems follows a similar philosophy, describing decreasing rules through comparatively elaborate graph-theoretic and syntactic conditions.

Our approach is complementary. Rather than beginning with a graph-theoretic notion of structural decrease, we start from elementary forms of string rewriting, and then investigate which extensions to term rewriting arise naturally. This bottom-up perspective may help distinguish structural conditions that are genuinely necessary for decidability from conditions introduced mainly to support a particular proof technique. In particular, the existing definition of contracting TRS is technically intricate, and it is not yet clear whether its complexity is essential to obtain decidability or whether simpler structural principles would suffice.

{\em \bf Organization of the paper.}
The rest of the paper is organised as follows. Section~\ref{sec:prelim} recalls some basic concepts in term rewriting and the deduction problems in the context of security protocol verification. In Section~\ref{sec:STS-erasing}, we explain how deduction problems modulo term rewriting systems reduce, in the unary setting, to right-divisibility problems for semi-Thue systems. 
Section~\ref{sec:prefix-erasing} and Section~\ref{sec:suffix-erasing} contain the main technical results of the paper. We prove that the right-divisibility problem is decidable for convergent prefix-erasing semi-Thue systems (Section~\ref{sec:prefix-erasing}) and for convergent suffix-erasing semi-Thue systems (Section~\ref{sec:suffix-erasing}), the latter also being known as dwindling systems. We also discuss the dual left-divisibility problem and explain why analogous decidability results hold for the same classes. 
In Section~\ref{sec:lifting-rules}, we discuss several variants and extensions of the problems studied in Section~\ref{sec:prefix-erasing} and Section~\ref{sec:suffix-erasing}.
Section~\ref{sec:conclusion} concludes the paper and presents directions for further research.

\begin{table*}[!t]
\centering
\footnotesize
\renewcommand{\arraystretch}{1.25}
\setlength{\tabcolsep}{3.5pt}

\newcolumntype{C}[1]{>{\centering\arraybackslash}m{#1}}
\newcolumntype{Y}{>{\centering\arraybackslash}X}

\newcommand{\holds}{%
    \textcolor{green!50!black}{\ensuremath{\checkmark}}}
\newcommand{\fails}{%
    \textcolor{red!75!black}{\ensuremath{\times}}}

\newcommand{\decidable}{%
    \textcolor{green!50!black}{\textbf{Dec}}}
\newcommand{\undecidable}{%
    \textcolor{red!75!black}{\textbf{Undec}}}

\begin{tabularx}{\textwidth}{C{2.65cm} *{7}{Y}}
\toprule
\rowcolor{gray!18}
\textbf{System class}
&
\multicolumn{5}{c}{\textbf{Structural properties}}
&
\makecell{\textbf{Word}\\\textbf{problem}}
&
\makecell{\textbf{RDiv}\\\textbf{problem}}
\\
\cmidrule(lr){2-6}

&
\textbf{Noeth.}
&
\makecell{\textbf{Length-}\\\textbf{reducing}}
&
\textbf{Monadic}
&
\textbf{Special}
&
\textbf{CR}
&
&
\\
\midrule

\makecell{Non-Noeth. \ \cite{jantzen2012confluent}}
& \fails
& \fails
& \fails
& \fails
& \holds
& \undecidable
& \undecidable
\\

\addlinespace[1pt]
\makecell{Len. Red. \ \cite{otto1984some}}
& \holds
& \holds
& \fails
& \fails
& \holds
& \decidable
& \undecidable
\\

\addlinespace[1pt]
\makecell{Mon. \ \cite{book1983decidable}}
& \holds
& \holds
& \holds
& \fails
& \holds
& \decidable
& \decidable
\\

\addlinespace[1pt]
\makecell{Special \ \cite{otto1991decision}}
& \holds
& \holds
& \holds
& \holds
& \holds
& \decidable
& \decidable
\\

\bottomrule
\end{tabularx}

\begin{minipage}{\textwidth}
\vspace{1.5mm}
\scriptsize
\textbf{Abbreviations:}
Noeth.\ = Noetherian;
CR = Church--Rosser.

\textbf{Structural properties:}
\holds\ means that the property is assumed;
\fails\ means that the property is not required.

\textbf{Decision status:}
\decidable\ = decidable;
\undecidable\ = undecidable;
\end{minipage}

\caption{Decidability of the word and right-divisibility problems for
selected classes of semi-Thue systems.}
\label{tab:word-rdiv-comparison}
\end{table*}

\section{Preliminaries}
\label{sec:prelim}

We recall the basic notions of term rewriting, string rewriting, frames and (intruder) deduction. For further details, we refer to \cite{baader1998term,book1993string,abadi2006deciding}.

\subsection{Term Rewriting}
Let $\mathcal{F}$ be a signature in which each symbol $f\in\mathcal{F}$ has an arity $\operatorname{ar}(f)\in\mathbb{N}$. Let $\mathcal{N}$ and $\mathcal{X}$ be disjoint countably infinite sets of names and variables, respectively, where names are treated as constants. We write $\mathcal{T}(\mathcal{F},\mathcal{N}\cup\mathcal{X})$ for the set of terms and $\mathcal{T}(\mathcal{F},\mathcal{N})$ for the set of ground terms. For a term $t$, the sets of variables, names, and subterms occurring in $t$ are denoted by $\operatorname{Var}(t)$, $\operatorname{Name}(t)$, and $\operatorname{Sub}(t)$, respectively.

A context $C[\circ_1,\ldots,\circ_m]$ is a term containing distinguished holes. Given terms $t_1,\ldots,t_m$, we write $C[t_1,\ldots,t_m]$ for the term obtained by replacing each hole $\circ_i$ by $t_i$. A substitution is a finitely supported mapping $\sigma\colon\mathcal{X}\to\mathcal{T}(\mathcal{F},\mathcal{N}\cup\mathcal{X})$, extended homomorphically to all terms. Its domain is $\operatorname{dom}(\sigma)=\{x\in\mathcal{X}\mid x\sigma\neq x\}$. A substitution is ground if $x\sigma$ is ground for every $x\in\operatorname{dom}(\sigma)$.

An equational theory $E$ is a set of equations between terms. The congruence generated by $E$ is denoted by $\equiv_E$; it is the smallest equivalence relation containing every substitution instance of every equation in $E$ and closed under contexts.

A term rewrite system (TRS) $R$ is a set of rules $\ell\to r$ such that $\ell\notin\mathcal{X}$ and $\operatorname{Var}(r)\subseteq\operatorname{Var}(\ell)$. The one-step rewrite relation induced by $R$ is defined by $C[\ell\sigma]\to_R C[r\sigma]$ for every rule $\ell\to r\in R$, context $C[\circ]$, and substitution $\sigma$. We denote the reflexive and transitive closure of $\to_R$ by $\to_R^*$ and its reflexive, symmetric, and transitive closure by $\xleftrightarrow{*}_R$.

A term $t$ is irreducible if there is no $u$ such that $t\to_R u$. If $t\to_R^*u$ and $u$ is irreducible, then $u$ is a normal form of $t$. A rewrite system is terminating if it admits no infinite rewrite sequence, and confluent if, whenever $t\to_R^*u$ and $t\to_R^*v$, there exists $w$ such that $u\to_R^*w$ and $v\to_R^*w$. It is convergent if it is both terminating and confluent. In this case, every term $t$ has a unique normal form, denoted by $t^\downarrow$, and $s\xleftrightarrow{*}_R t$ if and only if $s^\downarrow=t^\downarrow$.

A TRS $R$ presents an equational theory $E$ if $s\equiv_E t$ holds exactly when $s\xleftrightarrow{*}_R t$. Hence, if $R$ is convergent and effectively presented, equality modulo $E$ can be decided by comparing normal forms.

\subsection{String Rewriting}

Let $\Gamma$ be a finite alphabet. The set $\Gamma^*$ of finite words over $\Gamma$, equipped with concatenation and the empty word $\varepsilon$, is the free monoid generated by $\Gamma$. The length of a word $u$ is denoted by $|u|$. A word $v$ is a prefix of $u$ if $u=vw$, a suffix if $u=wv$, and a factor if $u=w_1vw_2$, for some $w,w_1,w_2\in\Gamma^*$. It is a scattered subword, or subsequence, of $u$ if it can be obtained from $u$ by deleting zero or more letters without changing the order of the remaining letters.

A semi-Thue system (STS), also called a string rewriting system, is a pair $S=(\Gamma,R)$, where $R\subseteq\Gamma^*\times\Gamma^*$. A rule $(\ell,r)\in R$ is written as $\ell\to r$. The induced one-step rewrite relation is given by $p\ell q\to_S prq$ for all $p,q\in\Gamma^*$ and $\ell\to r\in R$. Thus, an STS is the word-rewriting analogue of a TRS.
We write $\operatorname{IRR}(S)$ for the set of irreducible words.

The congruence generated by $S$ is defined by $u\equiv_S v$ if and only if $u\xleftrightarrow{*}_S v$. The quotient $M_S=\Gamma^*/{\equiv_S}$ is the monoid presented by $S$, with multiplication $[u]_S[v]_S=[uv]_S$. Equivalently, $M_S$ is generated by $\Gamma$ subject to the relations $\ell=r$ for every rule $\ell\to r\in R$. If $S$ is convergent, then $u\equiv_S v$ if and only if $u^\downarrow=v^\downarrow$.

A rule $\ell\to r$ is length-reducing if $|\ell|>|r|$, monadic if $|r|\leq 1$, and special if $r=\varepsilon$. An STS has one of these properties if every rule in it does. Every length-reducing STS is terminating.

\subsection{Frames and Deduction}

A \emph{frame} represents the messages available to an attacker. It is written as $\varphi=\nu\widetilde{n}.\sigma$, where $\widetilde{n}\subseteq\mathcal{N}$ is a finite set of restricted names and $\sigma$ is a finite ground substitution. We define $\operatorname{bn}(\varphi)=\widetilde{n}$, $\operatorname{dom}(\varphi)=\operatorname{dom}(\sigma)$, and $\operatorname{ran}(\varphi)=\{x\sigma\mid x\in\operatorname{dom}(\sigma)\}$.
For example, $\varphi=\nu k.\{x\mapsto\operatorname{enc}(m,k),y\mapsto h(k)\}$ has restricted name $k$, domain $\{x,y\}$, and range $\{\operatorname{enc}(m,k),h(k)\}$.

A \emph{recipe} for a frame $\varphi=\nu\widetilde{n}.\sigma$ is a term $M\in\mathcal{T}(\mathcal{F},(\mathcal{N}\setminus\widetilde{n})\cup\operatorname{dom}(\varphi))$. Thus, a recipe may use function symbols, unrestricted names, and handles from the frame, but cannot directly use restricted names. Its evaluation is the ground term $M\sigma$.
A ground term $t$ is \emph{deducible from $\varphi$ modulo $E$}, written $\varphi\vdash_E t$, if there exists a recipe $M$ such that $M\sigma\equiv_E t$. We may write $M:\varphi\vdash_E t$ when $M$ witnesses this deduction.

\begin{problemstatement}[Deduction Problem]
Given an equational theory $E$, a frame $\varphi=\nu\widetilde{n}.\sigma$, and a ground term $t$, decide whether $\varphi\vdash_E t$.
\end{problemstatement}

For example, consider $E=\{\operatorname{dec}(\operatorname{enc}(x,y),y)=x\}$ and the frame
$\varphi_1=\nu m_1,m_2,k_1,k_2.$ $\{x_1\mapsto\operatorname{enc}(m_1,k_1),x_2\mapsto\operatorname{enc}(m_2,k_2),x_3\mapsto k_1\}$.
Although $k_1$ is restricted, it is deducible using the handle $x_3$, since $x_3\sigma=k_1$. Moreover, $\operatorname{dec}(x_1,x_3)$ is a recipe for $m_1$, because $\operatorname{dec}(x_1,x_3)\sigma\equiv_E m_1$. Hence, $\varphi_1\vdash_E k_1$ and $\varphi_1\vdash_E m_1$. By contrast, $\varphi_1\not\vdash_E k_2$ and $\varphi_1\not\vdash_E m_2$.

\subsection{Right Divisibility Problem}

Suppose that $\Sigma$ consists of finitely many unary function symbols. We identify each function symbol in $\Sigma$ with a letter in a finite alphabet $\Gamma$. For each word $w=w_1\cdots w_n\in\Gamma^*$, let $\overline{w}(t)=w_1(w_2(\cdots w_n(t)\cdots))$, with $\overline{\varepsilon}(t)=t$. Every term over $\Sigma$ is therefore of the form $\overline{w}(a)$ or $\overline{w}(x)$ for some $w\in\Gamma^*$, name $a$, or variable $x$.

Given a frame $\varphi$, every recipe is correspondingly of the form $\overline{w}(a)$ for some public name $a$, or $\overline{w}(x)$ for some $x\in\operatorname{dom}(\varphi)$. In the latter case, the recipe can use only one variable from the frame. Suppose that $x\varphi=\overline{u}(d)$ and that the target term is $t=\overline{v}(c)$. Applying the recipe $r=\overline{w}(x)$ gives $r\varphi=\overline{w}(\overline{u}(d))=\overline{wu}(d)$. Hence, $r$ deduces $t$ precisely when $\overline{wu}(d)\equiv_E\overline{v}(c)$.

In particular, when $c=d$, this condition reduces to $wu\equiv_S v$, where $S$ is the semi-Thue system corresponding to the unary equational theory. Thus, deduction in this setting leads naturally to the following right-divisibility problem.

\begin{problemstatement}[Right-Divisibility]
Given a semi-Thue system $S=(\Gamma,R)$ and words $u,v\in\Gamma^*$, decide whether there exists $w\in\Gamma^*$ such that $wu\equiv_S v$.
\end{problemstatement}

The right-divisibility problem is a classical decision problem for semi-Thue systems and finitely presented monoids. It is closely related to the word problem, which asks whether $u\equiv_S v$ for given words $u,v\in\Gamma^*$ \cite{book1993string}, as well as to Green's relations and other divisibility problems \cite{book1983decidable}.

Its dual, the left-divisibility problem, asks whether there exists $w\in\Gamma^*$ such that $uw\equiv_S v$. For unary signatures, left-divisibility can be viewed as a restricted form of equational matching. Its study is also relevant to problems such as simultaneous unification and second-order equational unification \cite{otto1998equational}.

\section{STS with Structured Erasing Rules}
\label{sec:STS-erasing}

In this section, we define formally the classes of semi-Thue systems of interest. These classes are defined via rewrite rules that generalise subterm rules from term rewriting.
Recall that a subterm rule has the form \(C[t]\to t\): it removes the surrounding context \(C[\cdot]\) and
keeps the subterm \(t\). In the unary setting, terms can be identified with
words, where the outer context appears as a prefix and the inner subterm
appears as the remaining suffix  
Thus, a unary subterm rule corresponds to a
string rule which deletes a nonempty prefix and preserves the remaining
suffix. 
\begin{definition}
An STS \(S=(\Gamma,R)\) is called \emph{prefix-erasing} if, for all
\((u,v)\in R\), there exist \(x,y\in \Gamma^*\) such that \(|x|>0\), \(u=xy\), and
\(v=y\).
\end{definition}
Every prefix-erasing STS is terminating because it is length-reducing. However, it is not necessarily confluent.
\begin{example}\label{ex:prefix-erasing}
Let $\Gamma=\{a,b\}$ and $\Gamma'=\{a,b,c\}$. The following are examples of prefix-erasing STSs. 
\begin{enumerate}
\item Let $S_1=(\Gamma',R_1)$, where
$R_1=\{(bba, ba),\ (ccb, cb),\ (ac, c),$ $(bc, c),\ (ccc, cc)\}$.
This is a non-monadic prefix-erasing STS. It can be checked that $S_1$ is convergent.

\item Let $S_2=(\Gamma,R_2)$, where
$R_2=\{(bab, ab), (aa, \varepsilon), (bb, b)\}$.
This is also a non-monadic prefix-erasing STS. It can be checked that $S_2$ is convergent.

\item Let $S_3=(\Gamma,R_3)$, where
$R_3=\{(ab,b),\ (ba,a)\}$.
This system is prefix-erasing and terminating, but not confluent. Indeed,
$aba\to ba\to a$, while $aba\to aa$.
Since $a$ and $aa$ are distinct irreducible words, $S_3$ is not convergent.
\end{enumerate}

\end{example}

A natural dual to the prefix-erasing rules are the suffix erasing rules, e.g., a rule such as \(xy\to x\) that removes a nonempty suffix while preserving the prefix. Although
this is not a subterm rule in the usual sense, it is closely related to the
same intuition of context deletion, now retaining the outer context rather
than the inner component. 

\begin{definition}
An STS \(S=(\Gamma,R)\) is called \emph{suffix-erasing} if, for all
\((u,v)\in R\), there exist \(x,y\in \Gamma^*\) such that \(|y|>0\), \(u=xy\) and
\(v=x\).
\end{definition}

\begin{example}[Suffix-erasing]
\label{ex:suffix-erasing}
The following are the dual to the prefix-erasing systems in Example~\ref{ex:prefix-erasing}.
This reversal operation preserves termination and confluence (see, e.g., \cite{book1993string}). Hence we obtain, for example, the following convergent suffix-erasing STS from Example~\ref{ex:prefix-erasing} $S_1^R=(\Gamma',R_1^R)$, where
\(
R_1^R=\{(abb, ab),\ (bcc, bc),\ (ca, c),$ $\ (cb, c),\ (ccc, cc)\}.
\).
\end{example}

Another interesting erasing rule that might be considered is factor-erasing STS as follows.

\begin{definition}
An STS $S=(\Gamma,R)$ is called \emph{factor-erasing} if for every rule $(u,v)\in R$, there exist $x,y,z\in \Gamma^*$ with $|y|>0$ such that $u=xyz$ and $v=xz$.
\end{definition}

We conjecture that right divisibility problem is undecidable for factor erasing STS.

\section{Right divisibility problem is linear-time decidable for finite convergent prefix-erasing STS}
\label{sec:prefix-erasing}

In this section, we extend the known decidability results for right-divisibility in semi-Thue systems. Beyond convergent monadic systems, we show decidability for convergent prefix-erasing systems, whose right-hand sides may have arbitrary length.

Under the correspondence between words and unary terms, prefix-erasing rules are precisely subterm rules, and right-divisibility corresponds to deduction. Thus, this result follows from the decidability of deduction for convergent subterm theories \cite{abadi2006deciding}. Our contribution here is to make this connection explicit. In the next section, we will establishe a genuinely new result for convergent suffix-erasing systems.

    Our main tool is the set of suffix multiples.

\begin{definition}[Suffix multiples]
Let $S=(\Gamma,R)$ be an STS, and let $u \in \operatorname{IRR}(S)$. The set of \emph{$S$-suffix multiples} of $u$ is defined by $\operatorname{SufMul}_S(u)=\{v \in \operatorname{IRR}(S) : \exists p \in \Gamma^* \text{ such that } u=pv \text{ and } u \preceq_S^r v\}$.
\end{definition}
\begin{remark}
The set \(\operatorname{SufMul}_S(u)\) can be viewed as a word-based analogue of the saturated set $\operatorname{sat}(\varphi)$ used for subterm-convergent theories \cite{abadi2006deciding}. 
\end{remark}

We record some basic properties of $\operatorname{SufMul}_S(u)$. By definition, $u \in \operatorname{SufMul}_S(u)$, so the set is nonempty. Moreover, every $v \in \operatorname{SufMul}_S(u)$ is a suffix of $u$. Hence $\operatorname{SufMul}_S(u)$ is finite, with at most $|u|+1$ elements. Therefore, $\operatorname{SufMul}_S(u)$ has a unique shortest element, denoted by $\operatorname{MinSufMul}_S(u)$, or simply $u_{\min}$ when the context is clear. This element will be used to characterize right-divisibility by $u$.
Indeed, if
$u=u_1\cdots u_k u_{\min}$, where $u_1,\ldots,u_k \in \Gamma$, then
\(
\operatorname{SufMul}_S(u) =
\{u_j\cdots u_k u_{\min} : 1\leq j\leq k+1\},
\)
where $u_{k+1}\cdots u_k=\varepsilon$. Thus, $\operatorname{SufMul}_S(u)$ is completely determined by $u_{\min}$.

\begin{lemma}[Minimal suffix generator]\label{lem:minimal-suffix-generator}
Let \(S=(\Gamma,R)\) be a convergent prefix-erasing STS, and let \(u \in \operatorname{IRR}(S)\). Let \(u_{\min}\) be the unique element of \(\operatorname{SufMul}_S(u)\) of minimum length. Then, for every \(v \in \operatorname{IRR}(S)\), we have
$
u \preceq_S^r v
\quad \text{iff} \quad
u_{\min} \text{ is a suffix of } v.
$
\end{lemma}

We want to compute the element $u_{\operatorname{min}}$ of $\operatorname{SufMul}_S(u)$ with smallest length.

\begin{lemma}\label{lem:onealphabetremoval}
Let $u=cu_1\in \operatorname{IRR}(S)$, where $c\in\Gamma$ and
$u_1\in\Gamma^*$. Then $cu_1\preceq_S^r u_1$ if and only if there exist a rule $xy\to y$ in $R$ and words
$w_2\in\Gamma^+$ and $z,t\in\Gamma^*$ such that $x=w_2cz$ and $u_1=zyt$.
Consequently, the condition $cu_1\preceq_S^r u_1$ is decidable.
\end{lemma}

\begin{algorithm}[!t] \caption{Testing whether $cu_1\preceq_S^r u_1$} \label{alg:alphabetrightdivtest} \begin{algorithmic}[1] \Require Irreducible $cu_1$, with $c\in\Gamma$ \Ensure Whether $cu_1\preceq_S^r u_1$ \ForAll{$xy\to y\in R$ and factorisations $x=w_2cz$ with $w_2\in\Gamma^+$} \State \Return \textsc{Yes} \textbf{if} $zy$ is a prefix of $u_1$ \EndFor \State \Return \textsc{No} \end{algorithmic} \end{algorithm}

The preceding characterization yields the following decision procedure for testing whether $cu_1\preceq_S^r u_1$.

Algorithm~\ref{alg:alphabetrightdivtest} is effective because $R$ is finite and every left-hand side $x$ has only finitely many factorisations.

\begin{lemma}\label{lem:recursiveminsufmul}
For every $u\in \operatorname{IRR}(S)$, the word $u_{\min}$ is computable. More precisely,
$
\operatorname{MinSufMul}_S(\varepsilon)=\varepsilon,
$
and, for $u=cu_1$ with $c\in\Gamma$,
\[
\operatorname{MinSufMul}_S(cu_1)=
\begin{cases}
\operatorname{MinSufMul}_S(u_1),
    & \text{if } cu_1\preceq_S^r u_1,\\[1mm]
cu_1,
    & \text{otherwise.}
\end{cases}
\]
\end{lemma}

\begin{algorithm}[!t]
\caption{Computing $\operatorname{MinSufMul}_S(u)$}
\label{alg:minsufmul}
\begin{algorithmic}[1]
\Require Irreducible $u\in\Gamma^*$
\Ensure The shortest word in $\operatorname{SufMul}_S(u)$
\If{$u=\varepsilon$}
    \State \Return $\varepsilon$
\EndIf
\State Write $u=cu_1$, with $c\in\Gamma$
\If{Algorithm~\ref{alg:alphabetrightdivtest} returns \textsc{Yes} on $cu_1$}
    \State \Return $\operatorname{MinSufMul}_S(u_1)$
\EndIf
\State \Return $u$
\end{algorithmic}
\end{algorithm}

We are now ready to establish the main result of this section. Although it follows as a special case of the general decidability result of \cite{abadi2006deciding}, it is worth stating and proving separately. From the perspective of string rewriting, the result gives a natural extension of the classical decidability of right divisibility for monadic semi-Thue systems.

Our specialized treatment also yields a sharper algorithmic result than the general procedure. Every prefix-erasing semi-Thue system is length-reducing, and normal forms with respect to a fixed finite convergent length-reducing system can be computed in linear time \cite{book1982confluent}. Once the normal forms of the two input words have been obtained, our right-divisibility test also requires only linear time. Consequently, the complete procedure runs in time linear in the combined length of the input words.

\begin{theorem}[Linear-Time Right Divisibility for Prefix-Erasing Systems]
\label{thm:rdiv-prefix-erasing}
Let $S$ be a fixed finite convergent prefix-erasing semi-Thue system. The right-divisibility problem for $S$ is decidable in time $O(|u|+|v|)$, where $u$ and $v$ are the input words.
\end{theorem}

\begin{proof}
Let $u,v\in\Gamma^*$. Since $S$ is convergent, the words $u$ and $v$ have unique normal forms $u^{\downarrow}$ and $v^{\downarrow}$, and $u\preceq_S^r v$ if and only if $u^{\downarrow}\preceq_S^r v^{\downarrow}$.

By Lemma~\ref{lem:recursiveminsufmul}, we can compute $(u^{\downarrow})_{\min}=\operatorname{MinSufMul}_S(u^{\downarrow})$. Lemma~\ref{lem:minimal-suffix-generator} then gives
\[
u^{\downarrow}\preceq_S^r v^{\downarrow}
\quad\Longleftrightarrow\quad
(u^{\downarrow})_{\min}\text{ is a suffix of }v^{\downarrow}.
\]
Hence right-divisibility is decidable by normalizing $u$ and $v$, computing $(u^{\downarrow})_{\min}$, and performing one suffix test.

Because $S$ is fixed and prefix-erasing, it is length-reducing, and normal forms can be computed in linear time using a precompiled rule matcher and a stack-based reduction procedure. The computation of $(u^{\downarrow})_{\min}$ examines each suffix position at most once, with constant work per position depending only on $S$, and the final suffix test is linear. Thus, the complete procedure runs in time $O(|u|+|v|)$. A detailed analysis is given in Appendix~\ref{app:rdiv-complexity}.
\end{proof}

\begin{example}
Consider the STS $S_2=(\{a,b\},\{(bab,ab),\ (aa,\varepsilon),\ (bb, b)\})$ from Example~\ref{ex:prefix-erasing}.

We first determine whether $aba\preceq^r_{S_2}b$. Since both $aba$ and $b$ are irreducible, it suffices to compute $\operatorname{MinSufMul}_{S_2}(aba)$ and check whether it is a suffix of $b$.

Using rule $aa \to \varepsilon$, we write $x=aa$ and $y=\varepsilon$. 
Then $x=w_2az$ with $w_2=a$ and $z=\varepsilon$. Moreover, the condition $\varepsilon=zyt$ holds with $t=\varepsilon$, since $z=y=\varepsilon$. hence $aba \preceq_{S_2}^r ba$. Indeed, $a\cdot aba=aaba\to_{S_2}ba$. We next check whether the suffix $ba$ can be reduced further to $a$ in the right-divisibility ordering.

For the rule $bab\to ab$, write its left-hand side as $xy$ with $x=b$ and $y=ab$. The required condition would be $x=w_2bz$ for some $w_2\in\Gamma^+$ and $z\in\Gamma^*$. This is impossible because $x=b$ has length one.
Similarly, for the rule $bb\to b$, taking $x=b$ and $y=b$, the same condition $x=w_2bz$ cannot hold with $w_2\in\Gamma^+$.
Thus, $ba\not\preceq^r_{S_2}a$, and therefore $\operatorname{MinSufMul}_{S_2}(aba)=ba$. Since $ba$ is not a suffix of $b$, we conclude that $aba\not\preceq^r_{S_2}b$.

As a second illustration, consider whether $a\preceq^r_{S_2}ab$. We compute $\operatorname{MinSufMul}_{S_2}(a)$. Similar to previous use of $aa \to \varepsilon$, we have $a \preceq_{S_2}^r \varepsilon$, and consequently $\operatorname{MinSufMul}_{S_2}(a)=\varepsilon$.
Since $\varepsilon$ is a suffix of every word, in particular of $ab$, it follows that $a\preceq^r_{S_2}ab$. Indeed, a witness is the word $aba$, because $aba\cdot a=abaa\to_{S_2}ab$.
\end{example}

By the duality of the left-divisibility and right-divisibility problems,  we have the following theorem.
\begin{corollary}
    The left-divisibility problem is decidable for every finite convergent suffix-erasing STS.
\end{corollary}

For suffix-erasing systems, the only if direction of Lemma~\ref{lem:minimal-suffix-generator} is false. Consider the STS $S=(\Gamma,R)$ with $\Gamma=\{a,b\}$ and $R=\{(ba,b)\}$. This system is suffix-erasing, since the rule $ba\to b$ deletes the nonempty suffix $a$. It is also convergent: the rule strictly decreases word length, and there are no nontrivial critical overlaps.

Now take $u=a$ and $v=b$. Both words are irreducible. We have $u\preceq_S^r v$, because there exists $w=b$ such that $wu=ba\to b=v$. However, $\operatorname{SufMul}_S(a)=\{a\}$, so $\operatorname{SufMul}_{S_1}(a)$. Since $a$ is not a suffix of $b$, the condition that $u_{\operatorname{min}}$ is a suffix of $v$ fails.

Thus, for this convergent suffix-erasing STS, we have $u\preceq_S^r v$ but $u_{\operatorname{min}}$ is not a suffix of $v$. Therefore Lemma~\ref{lem:minimal-suffix-generator} is genuinely specific to the prefix-erasing setting. This is precisely where the forward-search algorithm breaks down for suffix-erasing systems: even after computing $u_{\operatorname{min}}$, checking whether it occurs as a suffix of $v$ no longer characterises right-divisibility.

\section{Right divisibility problem is decidable for finite convergent suffix-erasing STS}\label{sec:suffix-erasing}

In this section, we study the right divisibility for the dual of the class of prefix-erasing STS, namely suffix-erasing STS.
We present an algorithm that decides this in complexity that is exponential in the length of $|u|$ in the worst case, if we want to decide whether $u \preceq_S^r v$ where $S$ is a finite convergent suffix-erasing STS.

Suppose without loss of generality, $u$ and $v$ are already in normal form.
Let $u=u_1u_2 \dots u_n$. Note that if $u \preceq_S^l v$, then it must be the case that $u_k u_{k+1} \dots u_n \preceq_S^r v$ for every $k=1,2,\dots,n$. 

Let $c \in \Gamma$. We define a relation $R_c=\{(u,v) \in \operatorname{IRR}(S) \times \operatorname{IRR}(S):uc \equiv_S v\}$. The graph $(\operatorname{IRR}(S),\{R_c\}_{c \in \Gamma})$ is Cayley graph of monoid $\Gamma^*/\equiv_S$ represented by STS $S$. 

We define a function $\delta:\Gamma \times \operatorname{IRR}(S) \to {\mathcal P}(\operatorname{IRR}(S))$ such that for every $c \in \Gamma$ and $u \in \operatorname{IRR}(S)$, $v \in \delta(c,u)$ if and only if $vc \equiv_S u$. Cayley graphs of monoids have been studied in several settings. However, as far as we are aware, our result does not follow directly from the existing literature \cite{caucal2020cayley,kuske2006logical}. In particular, the structure of Cayley graphs of monoids presented by convergent suffix-erasing semi-Thue systems appears not to have been investigated. We hope that this observation will motivate further study of monoids presented by broader classes of erasing semi-Thue systems.

\begin{definition}\label{def:delta}
    Let $S$ be a convergent STS. For $c\in \Gamma$ and $u\in \operatorname{IRR}(S)$, define
    \[
        \delta(c,u)
        =
        \{v\in \operatorname{IRR}(S) : vc\equiv_S u\}.
    \]
\end{definition}

\begin{lemma}\label{lem:backward-delta}
    Let $S$ be a convergent STS, and let $u,v\in \operatorname{IRR}(S)$. Write $u=u_1u_2\cdots u_n$, where $u_i\in \Gamma$. Define $B_{n+1}=\{v\}$ and, for $k=n,n-1,\dots,1$, define
    \[
        B_k=\bigcup_{z\in B_{k+1}}\delta(u_k,z).
    \]
    For every $k=1,2,\dots,n+1$, we have
    \[
        B_k^S(u,v)
        =
        \{x\in \operatorname{IRR}(S) : xu_ku_{k+1}\cdots u_n\equiv_S v\},
    \]
    where $u_{n+1}u_{n+2}\cdots u_n$ is understood as $\varepsilon$.
\end{lemma}

\begin{corollary}\label{cor:rightdivisibility-b1}
    With the notation above,
    we have
$u\preceq^r_S v$ if and only if 
$B_1^S(u,v)\neq \varnothing$.
\end{corollary}

\begin{lemma}\label{lem:finite-delta}
Let $S=(\Gamma,R)$ be a convergent suffix-erasing STS, where $R$ is finite and every rule has the form $xy \to x$ with $y \neq \epsilon$. Then, for all $c \in \Gamma$ and $u \in \operatorname{IRR}(S)$, the set $\delta(c,u)$ is finite.
\end{lemma}

\begin{theorem}
The right-divisibility problem is decidable for every finite convergent suffix-erasing STS.
\end{theorem}

\begin{proof}
Let $S$ be a finite convergent suffix-erasing STS, and let $u,v \in \Gamma^*$ be the input words. We want to decide whether $u \preceq_S^r v$.

Since $S$ is convergent, every word has a unique normal form, and this normal form is computable. Moreover, right-divisibility is invariant under replacing words by equivalent words. Therefore, we may first replace $u$ and $v$ by their normal forms. Thus, without loss of generality, assume that $u,v \in \operatorname{IRR}(S)$.

Write $u=u_1u_2\cdots u_n$, where each $u_i \in \Gamma$. Starting from $v$, we compute the sets $B_{n+1},B_n,\dots,B_1$ backwards as follows:
$
B_{n+1}=\{v\},
$
and, for each $k=n,n-1,\dots,1$,
\[
B_k=\bigcup_{z \in B_{k+1}}\delta(u_k,z).
\]

By Lemma~\ref{lem:finite-delta}, for every $c \in \Gamma$ and every irreducible word $z$, the set $\delta(c,z)$ is finite. Since $S$ is finite and suffix-erasing, the proof of Lemma~\ref{lem:finite-delta} gives an effective way to compute $\delta(c,z)$. Hence, starting from the finite set $B_{n+1}=\{v\}$, we can effectively compute each finite set $B_n,B_{n-1},\dots,B_1$.

By Corollary~\ref{cor:rightdivisibility-b1}, we have
$
u \preceq_S^r v
\quad\text{if and only if}\quad
B_1 \neq \varnothing.
$
Thus, after computing $B_1$, it remains only to check whether this finite set is empty. This gives a decision procedure for the right-divisibility problem.
\end{proof}

\begin{algorithm}[!t]
\caption{Right-divisibility for convergent suffix-erasing STSs}
\label{alg:rightdiv-suffix}
\begin{algorithmic}[1]
\Require A finite convergent suffix-erasing STS $S=(\Gamma,R)$ and
         $u,v\in\Gamma^*$
\Ensure Whether $u\preceq_S^r v$
\State Compute the normal forms $\widehat{u}$ and $\widehat{v}$
\If{$\widehat{u}=\varepsilon$}
    \State \Return \textsc{Yes}
\EndIf
\State Write $\widehat{u}=u_1\cdots u_n$, with each $u_i\in\Gamma$
\State $B_{n+1}\gets\{\widehat{v}\}$
\For{$k=n,n-1,\ldots,1$}
    \State $B_k\gets\displaystyle\bigcup_{z\in B_{k+1}}\delta(u_k,z)$
\EndFor
\State \Return \textsc{Yes} if $B_1\neq\varnothing$; otherwise \textsc{No}
\end{algorithmic}
\end{algorithm}

\begin{example}
Consider the STS $S=(\Gamma,R)$ with $\Gamma=\{a,b,c\}$ and $$R=\{(abb,ab),(bcc,bc),(ca,c),(cb,c),(ccc,cc)\}.$$ 
    Now, we want to check two right-divisbility problems:
    \begin{enumerate}
        \item[(i)] Deciding whether $abc \preceq_S^r bac$.

        Both \(abc\) and \(bac\) are irreducible. Write \(abc=u_1u_2u_3\), where \(u_1=a\), \(u_2=b\), and \(u_3=c\). We start with \(B_4=\{bac\}\). Then we compute backwards. First, \(B_3=\delta(c,bac)\). Since \(bac=ba \cdot c\), we have \(ba \in \delta(c,bac)\). There are no other possibilities, so \(B_3=\{ba\}\).

        Next, \(B_2=\delta(b,ba)\). There is no irreducible word \(x\) such that \(xb \equiv_S ba\). Indeed, \(ba\) does not end in \(b\), and none of the rules can reduce a word of the form \(xb\) to \(ba\). Hence \(B_2=\varnothing\). Therefore \(B_1=\varnothing\).
        
        By the backward-divisibility criterion, we have \(abc \preceq_S^r bac\) if and only if \(B_1\) is non-empty. Since \(B_1=\varnothing\), we conclude that \(abc \npreceq_S^r bac\). Equivalently, there is no word \(w \in \Gamma^*\) such that \(wabc \equiv_S bac\).
        
        \item[(ii)] Deciding whether $bb \preceq_S^r ab$.
        
        Both \(bb\) and \(ab\) are irreducible. Write \(bb=u_1u_2\), where \(u_1=b\) and \(u_2=b\). We start with \(B_3=\{ab\}\). Then \(B_2=\delta(b,ab)\). Now there are two possibilities. First, since \(ab=a \cdot b\), we have \(a \in \delta(b,ab)\). Second, since \(abb \to ab\), we also have \(ab \in \delta(b,ab)\). Hence \(B_2=\{a,ab\}\). 
        
        Next, \(B_1=\delta(b,a)\cup \delta(b,ab)\). We have \(\delta(b,a)=\varnothing\), while \(\delta(b,ab)=\{a,ab\}\). Therefore \(B_1=\{a,ab\}\). Since \(B_1\neq\varnothing\), the backward-divisibility criterion implies that \(bb \preceq_S^r ab\). Indeed, there are witnesses: taking \(w=a\), we get \(abb \to ab\), and taking \(w=ab\), we get \(abbb \to abb \to ab\).
    \end{enumerate}

\end{example}

The worst-case complexity of this algorithm is base-$|R|$ exponential time in the length of input and is discussed in Appendix~\ref{app:backwardcomplexity}

It is not hard to using properties reversal of STS, that the following statement is true.

\begin{corollary}
    The left-divisibility problem is decidable for every finite convergent prefix-erasing STS.
\end{corollary}

For prefix-erasing systems, however, $\delta(c,u)$ need not be finite. Consider the convergent prefix-erasing STS $S=(\{a,b\},R)$ with $R=\{(ab,b)\}$. The rule $ab\to b$ deletes the nonempty prefix $a$, so $S$ is prefix-erasing. 
Consider $\delta(b,b)$. For every $m\geq 0$, the word $a^m$ is irreducible, and $a^m b\to_S^* b$. Hence $a^m\in \delta(b,b)$ for every $m\geq 0$. Therefore $\delta(b,b)$ is infinite.

This shows precisely where the termination of backward-search algorithm breaks down for prefix-erasing systems. The equivalence $u\preceq_S^r v$ if and only if $B_1\neq \varnothing$ remains true, but the sets $B_k$ may be infinite, because the transition sets $\delta(c,s)$ may be infinite. Thus the naive backward construction no longer gives a finite decision procedure.

\section{Lifting Rules for Term Rewriting Systems}
\label{sec:lifting-rules}

\subsection{Lifting Rules}

The erasing classes of semi-Thue systems considered in the preceding sections suggest natural analogues for term rewriting systems over arbitrary signatures. In particular, the correspondences between prefix-, suffix-, and factor-erasing string rewrite rules motivate classes of term rewrite rules in which contexts surrounding variables or subterms are erased.

We begin with the standard notion of a subterm rule. A rule of this form can be written as $C[t]\to t$, where the right-hand side is obtained by selecting a subterm of the left-hand side and erasing its surrounding context. Over a unary signature, subterm rules correspond to prefix-erasing string rewrite rules. Abadi and Cortier showed that deduction and static equivalence are decidable for subterm-convergent equational theories \cite{abadi2006deciding}.

We introduce two related classes of rewrite rules. The first is the class of \emph{variable-lifting rules}, in which an inner context surrounding a variable is removed while an outer context is preserved. Over unary signatures, variable-lifting rules correspond to suffix-erasing STSs. The second and more general class is that of \emph{subterm-lifting rules}, in which an arbitrary term, rather than only a variable, is lifted through an erased context -- a more general analogue of factor-erasing STS.

\begin{definition}[Variable lifting]
A rewrite rule $l\to r$ is called a \emph{variable-lifting rule} if there exist one-hole contexts $C[\circ]$ and $D[\circ]$, with $D[\circ]\neq\circ$, and a variable $x$ such that $l=C[D[x]]$ and $r=C[x]$. A TRS is called \emph{variable-lifting} if each of its rules is a variable-lifting rule.
\end{definition}

\begin{definition}[Subterm lifting]
A rewrite rule $l\to r$ is called a \emph{subterm-lifting rule} if there exist one-hole contexts $C[\circ]$ and $D[\circ]$, with $D[\circ]\neq\circ$, and a term $t$ such that $l=C[D[t]]$ and $r=C[t]$. A TRS is called \emph{subterm-lifting} if each of its rules is a subterm-lifting rule.
\end{definition}

Every variable-lifting rule is a subterm-lifting rule, obtained by taking $t=x$. Moreover, subterm lifting generalizes the standard notion of a subterm rule. Indeed, if the outer context is the identity context $C[\circ]=\circ$, then a subterm-lifting rule has the form $D[t]\to t$ and is therefore a subterm rule.

The presence of the outer context $C$ nevertheless makes subterm-lifting theories behave differently from ordinary subterm theories with respect to deduction. The following example illustrates this distinction.

\begin{example}
Consider the equational theories $E=\{f(g(x))\approx f(x)\}$ and $E'=\{g(x)\approx x\}$, together with the frame $\varphi=\nu a.\{w\mapsto g(a)\}$. The theory $E$ is induced by a subterm-lifting rule, whereas $E'$ is induced by a subterm rule.

Under both theories, the term $f(a)$ is deducible. Indeed, applying the public symbol $f$ to the term represented by $w$ yields $f(g(a))$, and we have $f(g(a))=_E f(a)$ and $f(g(a))=_{E'}f(a)$. Hence, $\varphi\vdash_E f(a)$ and $\varphi\vdash_{E'}f(a)$.

The term $a$, however, is deducible only under $E'$. Under $E'$, the recipe $w$ evaluates to $g(a)\approx_{E'}a$, so $\varphi\vdash_{E'}a$. In contrast, the equation in $E$ removes $g$ only when it occurs immediately below $f$. It does not identify $g(a)$ with $a$, and therefore $\varphi\not\vdash_E a$.
\end{example}

Thus, although the rules $f(g(x))\to f(x)$ and $g(x)\to x$ erase the same inner context $g[\circ]$, the preserved outer context restricts where the erasure can be applied. Subterm lifting can consequently produce a different deduction relation from that induced by ordinary subterm rules.

\subsection{Decision problems}

The correspondence with erasing STSs motivates natural decision problems for these classes. We showed in Section~\ref{sec:suffix-erasing} that right-divisibility is decidable for convergent suffix-erasing STSs. Since suffix-erasing rules correspond, over unary signatures, to variable-lifting rules, it is natural to ask whether an analogous decidability result holds for deduction.

\begin{openquestion}
Is deduction decidable for convergent variable-lifting TRSs?
\end{openquestion}

Since variable lifting is a special case of subterm lifting, a more general question is the following.

\begin{openquestion}
Is deduction decidable for convergent subterm-lifting TRSs?
\end{openquestion}

Despite their simple syntactic form, variable-lifting theories do not necessarily satisfy the local stability property. Local stability is one of the standard techniques used to establish decidability of deduction and static equivalence, and it is used in the decidability results for convergent contracting TRSs in \cite{abadi2006deciding,bunch2024knowledge}. Thus, the existing local-stability framework does not immediately settle the preceding questions.

\begin{definition}[Local Stability {\cite{abadi2006deciding,bunch2024knowledge}}]
A convergent TRS, $R$, is \emph{locally stable} if, for every
$R$-normalized frame $\phi = \nu \widetilde{n}.\sigma$, there exists a
finite set $\operatorname{sat}(\phi)$ of ground terms such that:
\begin{itemize}
    \item $\operatorname{ran}(\sigma) \subseteq \operatorname{sat}(\phi)$
    and $\operatorname{fn}(\phi) \subseteq \operatorname{sat}(\phi)$;

    \item if $M_1,\ldots,M_k \in \operatorname{sat}(\phi)$ and
    $f(M_1,\ldots,M_k)
    \in \operatorname{st}(\operatorname{sat}(\phi))$, then
    $f(M_1,\ldots,M_k) \in \operatorname{sat}(\phi)$;

    \item if
    $C[S_1,\ldots,S_l] \xrightarrow{\epsilon}_{R} M$, where $C$ is a
    context with $|C| \leq c_R$ and
    $\operatorname{fn}(C) \cap \widetilde{n} = \varnothing$, and
    $S_1,\ldots,S_l \in \operatorname{sat}(\phi)$, then there exist a
    context $C'$ and $S'_1,\ldots,S'_k \in \operatorname{sat}(\phi)$
    such that $|C'| \leq c_R^2$,
    $\operatorname{fn}(C') \cap \widetilde{n} = \varnothing$, and
    $M \xrightarrow{*}_{R} C'[S'_1,\ldots,S'_k]$.
    \item if $M \in \operatorname{sat}(\phi)$, then $\phi \vdash_R M$.
\end{itemize}
\end{definition}

\begin{proposition}\label{prop:notlocallystable}
There exists a convergent variable-lifting theory $E$ and a frame $\varphi$ such that $\varphi$ is not $E-$local stabile.
\end{proposition}

\subsection{Simultaneous lifting}

We now generalize the preceding notions by allowing several non-overlapping inner contexts to be erased simultaneously.

\begin{definition}[Simultaneous variable lifting]
Let $n\geq 1$. A rewrite rule $l\to r$ is a \emph{simultaneous variable-lifting (SVL) rule} if there exist an $n$-hole context $C[\circ_1,\ldots,\circ_n]$, nontrivial one-hole contexts $D_1[\circ],\ldots,D_n[\circ]$, and pairwise distinct variables $x_1,\ldots,x_n$ such that
\begin{align*}
l &= C[D_1[x_1],\ldots,D_n[x_n]],\\
r &= C[x_1,\ldots,x_n].
\end{align*}
A TRS is called \emph{simultaneous variable-lifting} if each of its rules is an SVL rule.
\end{definition}

An SVL rule simultaneously erases the inner contexts $D_1,\ldots,D_n$ while preserving the common outer context $C$. Ordinary variable lifting is precisely the special case $n=1$.

\begin{definition}[Simultaneous subterm lifting]
Let $n\geq 1$. A rewrite rule $l\to r$ is called a \emph{simultaneous subterm-lifting rule}, or an \emph{SSL rule}, if there exist an $n$-hole context $C[\circ_1,\ldots,\circ_n]$, nontrivial one-hole contexts $D_1[\circ],\ldots,D_n[\circ]$, and terms $t_1,\ldots,t_n$ such that
\[
l = C[D_1[t_1],\ldots,D_n[t_n]] \qquad \mbox{ and } \qquad r = C[t_1,\ldots,t_n].
\]
A TRS is called \emph{simultaneous subterm-lifting} if each of its rules is an SSL rule.
\end{definition}

\begin{remark}
    Every SVL rule is an SSL rule, obtained by taking $t_i=x_i$ for each $i\in{1,\ldots,n}$.
\end{remark}

\begin{remark}
    The inclusion of SVL in SSL is strict: for example, $R=\{f(g(h(g(x))),g(y)) \to f(g(g(x)),y)\}$ is an SSL TRS, but not SVL. 
\end{remark}
 Thus, simultaneous subterm lifting generalizes simultaneous variable lifting by permitting arbitrary terms, rather than only pairwise distinct variables, to be lifted through the erased contexts. Ordinary subterm lifting is the special case $n=1$.

\subsection{Relation to embedded rewrite systems}

Our approach is complementary to previous extensions of subterm-convergent theories. Bunch et al.~\cite{bunch2024knowledge} introduced the classes of graph-embedded and contracting TRSs in order to capture equational theories lying beyond the subterm-convergent setting. They showed that deduction and static equivalence are decidable for convergent contracting TRSs, whereas both problems are undecidable for the broader class of convergent graph-embedded TRSs.

The classes studied here arise from a different perspective. Rather than beginning with a general embedding relation and identifying a subclass for which existing proof techniques apply, we begin with simple erasing string rewrite systems. Decidability results in the unary setting then motivate natural generalizations to arbitrary signatures. We next compare these lifting classes with the homeomorphic-embedding relation considered in \cite{bunch2024knowledge}.

\begin{definition}[Homeomorphic embedding {\cite{bunch2024knowledge}}]
The \emph{homeomorphic-embedding relation}, denoted by $\succeq_{\mathit{emb}}$, is the least binary relation on terms satisfying the following conditions:
\begin{enumerate}
\item $x\succeq_{\mathit{emb}}x$ for every variable $x$;

\item if $s_i\succeq_{\mathit{emb}}t_i$ for every $i\in\{1,\ldots,n\}$, then $f(s_1,\ldots,s_n)\succeq_{\mathit{emb}}f(t_1,\ldots,t_n)$;

\item if $s_i\succeq_{\mathit{emb}}t$ for some $i\in\{1,\ldots,n\}$, then $f(s_1,\ldots,s_n)\succeq_{\mathit{emb}}t$.

\end{enumerate}
A TRS $R$ is called \emph{homeomorphic-embedded} if $\ell\succeq_{\mathit{emb}}r$ for every rule $\ell\to r\in R$.
\end{definition}

The second condition preserves a common outer function symbol, whereas the third removes a surrounding context by selecting one of its arguments. These two operations are sufficient to embed the right-hand side of every simultaneous subterm-lifting rule into its left-hand side.

\begin{proposition}\label{prop:homemb}
Every SSL TRS is homeomorphic-embedded.
\end{proposition}
\begin{remark}
    There is a non-SSL homeomorphic-embedding TRS. For example, consider TRS $R=\{f(g(h(x)),y) \to g(x)\}$ is not SSL, but a homeomorphic embedding TRS.
\end{remark}

\begin{corollary}
    Every SSL and SVL (and hence also ordinary subterm lifting and variable lifting) TRSs are homeomorphic-embedding.
\end{corollary}

In \cite{bunch2024knowledge}, it has been proved that there exists homeomorphic embedded TRS for which deduction is undecidable. The simultaneous lifting classes therefore identify syntactically restricted subclasses of homeomorphic-embedded systems whose knowledge problems remain to be investigated.
This inclusion relations can be visualized as in the diagram in Figure~\ref{fig:TRS-relations}.

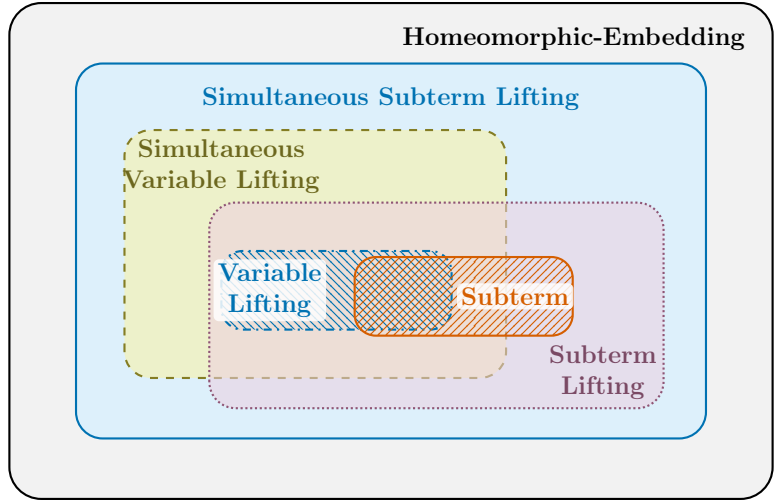
\begin{figure}[t]
\begin{center}
\definecolor{oiBlue}{HTML}{0072B2}       
\definecolor{oiSky}{HTML}{56B4E9}        
\definecolor{oiYellow}{HTML}{F0E442}     
\definecolor{oiPurple}{HTML}{CC79A7}     
\definecolor{oiVermillion}{HTML}{D55E00} 

\begin{tikzpicture}[scale=0.8, every node/.style={transform shape}]
    \draw[black, thick, rounded corners=12pt, fill=black!5]
        (-6.3,-3.7) rectangle (6.3,4.5);
    \node[anchor=north east] at (6.0,4.25) {\textbf{Homeomorphic-Embedding}};

    \draw[oiBlue, thick, rounded corners=10pt, fill=oiSky!20]
        (-5.2,-2.7) rectangle (5.2,3.5);
    \node[anchor=north, text=oiBlue] at (0,3.25)
        {\textbf{Simultaneous Subterm Lifting}};

    \draw[oiYellow!55!black, thick, dashed, rounded corners=10pt,
          fill=oiYellow!45, fill opacity=0.6]
        (-4.4,-1.7) rectangle (1.9,2.4);
    \node[align=center, text=oiYellow!45!black] at (-2.8,1.8)
        {\textbf{Simultaneous}\\\textbf{Variable Lifting}};

    \draw[oiPurple!70!black, thick, densely dotted, rounded corners=10pt,
          fill=oiPurple!35, fill opacity=0.55]
        (-3,-2.2) rectangle (4.5,1.2);
    \node[align=center, text=oiPurple!60!black] at (3.5,-1.6)
        {\textbf{Subterm}\\\textbf{Lifting}};

    \draw[oiBlue, thick, dash dot, rounded corners=8pt,
          pattern=north west lines, pattern color=oiBlue]
        (-2.8,-0.9) rectangle (1.0,0.4);
    \node[align=center, text=oiBlue, fill=white, fill opacity=0.85,
          text opacity=1, rounded corners=2pt, inner sep=2pt] at (-2,-0.25)
        {\textbf{Variable}\\\textbf{Lifting}};

    \draw[oiVermillion, thick, rounded corners=8pt,
          pattern=north east lines, pattern color=oiVermillion]
        (-0.6,-1.0) rectangle (3.0,0.3);
    \node[text=oiVermillion, fill=white, fill opacity=0.85,
          text opacity=1, rounded corners=2pt, inner sep=2pt] at (2.05,-0.35)
        {\textbf{Subterm}};
\end{tikzpicture}

\end{center}
\caption{Relations between different TRS classes.}
\label{fig:TRS-relations}
\end{figure}

We notice however that our notion of variable lifting theories are sharper to capture the TRS for which deduction becomes undecidable. The proof is by reduction to MPCP problem which is known to be undecidable. 

\begin{theorem}
There exists a convergent simultaneous variable-lifting TRS for which the deduction problem is undecidable.
\end{theorem}

\begin{proof}
    We reduce MPCP problem $\Pi=\{(\alpha_i,\beta_i):i=0,1,\dots,n\} \subseteq \Gamma^* \times \Gamma^*$ to deduction problem for a simultaneous subterm lifting theory. Consider a TRS with signature $\Sigma = \{f,locked\} \cup \{\overline{\alpha}: \alpha \in \Gamma\}$ where $f$ is a ternary function symbol, $locked$ is a unary function symbol, and one unary function symbol $\overline{\alpha}$ for every symbol $\alpha \in \Gamma$. For each $(\alpha_i,\beta_i)$ with $1 \le i \le n$, we add rules
    \begin{align*}
        f(\overline{\alpha_i}(x),\overline{\beta_i}(y),locked(z)) &\to f(x,y,locked(z))
    \end{align*}
    and, on top of that, the rule $f(x,x,locked(z)) \to f(x,x,z)$ in $R$.
    Each rule 
    $$f(\overline{\alpha_i}(x),\overline{\beta_i}(y),locked(z)) \to f(x,y,locked(z))$$ 
    can be written as $C[D_1[x],D_2[y]] \to C[x,y]$ where
    $$
    C[\circ_1, \circ_2] = f(\circ_1, \circ_2, locked(z)), \quad 
    D_1[\circ] = \overline{\alpha_i}(\circ), \quad \mbox{and} \quad 
    D_2[\circ] = \overline{\beta_i}(\circ).
    $$
    Finally, rule $f(x,x,locked(z)) \to f(x,x,z)$ can be written as $C[D[z]] \to C[z]$ where
    $$
    C[\circ] = f(x,x,\circ) \quad \mbox{and} \quad 
    D[\circ] = locked(\circ).
    $$
    Indeed, this rule is a simultaneous variable lifting rule with $n=2$. The correctness of reduction can be seen in \cite{ bunch2024knowledge}.
\end{proof}

\begin{corollary}
    There exists a convergent simultaneous subterm-lifting TRS for which the deduction problem is undecidable.
\end{corollary}

\section{Conclusion and Further Research}
\label{sec:conclusion}

Knowledge problems originate from security protocol verification, where one asks what an attacker can compute and distinguish. Although motivated by practice, these problems are fundamentally mathematical: their decidability depends on the algebraic structure of the underlying equational theory. Studying simplified classes, even when they are not directly motivated by cryptographic applications, helps isolate the structural sources of decidability and undecidability, and provides a principled path towards richer equational theories.

In this paper, we investigated knowledge problems through erasing semi-Thue systems and their term-rewriting analogues. 

Overall, our results demonstrate that even simple erasing structures exhibit a rich decidability landscape. We hope that understanding these fundamental classes will ultimately lead to decision procedures for more expressive equational theories arising in symbolic cryptographic protocol analysis.

We conclude with several directions for further work. The most immediate open questions are whether deduction is decidable for convergent variable-lifting and subterm-lifting TRSs over arbitrary signatures.  Static equivalence remains largely unexplored for classical classes of string-rewriting systems. Natural questions are whether it is decidable for convergent monadic, prefix-erasing, suffix-erasing, or factor-erasing STSs, and which properties of the quotient monoid $\Gamma^*/{\equiv_S}$ determine its complexity.
Lastly, the lifting rules considered here preserve a common outer context. More expressive rules may both erase contexts and move subterms between positions, for example $C[D[s],t] \to D[t]$.
Such behaviour occurs in cryptographic equations such as blind signatures:
\begin{align*}
\mathsf{unblind}
\bigl(
\mathsf{sign}(\mathsf{blind}(m,r),k),
r
\bigr)
\to
\mathsf{sign}(m,k).
\end{align*}
A possible generalization is to describe rules through mappings between positions on their two sides, with restrictions on duplication, cycles, and changes in depth. This may yield classes broader than subterm-lifting systems while retaining enough structure for normalization or locality arguments.

\bibliography{knowledge}

\appendix

\section{Proofs}

\subsection{Proof of Lemma~\ref{lem:minimal-suffix-generator}}

\begin{proof}
We first prove the forward direction. Suppose \(u \preceq_S^r v\). Then there exists \(w \in \Gamma^*\) such that \(wu \equiv_S v\). Since \(S\) is convergent and \(v\) is irreducible, we have \(wu \to_S^* v\).

We claim that there exists some \(r \in \operatorname{SufMul}_S(u)\) such that \(r\) is a suffix of \(v\). Indeed, along any reduction sequence starting from a word of the form \(wu\), prefix-erasing rules can only delete letters to the left of some suffix of \(u\), while preserving that suffix or replacing it by a shorter suffix which is still a right multiple of \(u\). Therefore the final irreducible word \(v\) must have some suffix \(r\) such that \(r\) is a suffix of \(u\) and \(u \preceq_S^r r\). Since \(v\) is irreducible, this suffix \(r\) is also irreducible, and hence \(r \in \operatorname{SufMul}_S(u)\).

Now \(u_{\min}\) and \(r\) are both suffixes of \(u\), and \(u_{\min}\) has minimum length among all elements of \(\operatorname{SufMul}_S(u)\). Since suffixes of a fixed word are linearly ordered by the suffix relation, \(u_{\min}\) is a suffix of \(r\). Since \(r\) is a suffix of \(v\), it follows that \(u_{\min}\) is also a suffix of \(v\).

Conversely, suppose that \(u_{\min}\) is a suffix of \(v\). Then there exists \(v' \in \Gamma^*\) such that \(v=v'u_{\min}\). Since \(u_{\min} \in \operatorname{SufMul}_S(u)\), we have \(u \preceq_S^r u_{\min}\). Hence there exists \(w \in \Gamma^*\) such that \(wu \equiv_S u_{\min}\). By congruence of \(\equiv_S\), we may add the same left context \(v'\) to both sides, obtaining \(v'wu \equiv_S v'u_{\min}=v\). Therefore there exists a word \(w'=v'w\) such that \(w'u \equiv_S v\). Hence \(u \preceq_S^r v\).
\end{proof}

\subsection{Proof of Lemma~\ref{lem:onealphabetremoval}}

\begin{proof}
Suppose first that $cu_1\preceq_S^r u_1$. Choose a shortest word
$w\in\Gamma^*$ such that $wcu_1\equiv_S u_1$. Since $S$ is
convergent and $u_1$ is irreducible, we have $wcu_1\to_S^*u_1$.

The minimality of $w$ implies that $w$ is irreducible: otherwise,
reducing a redex inside $w$ would produce a shorter witness. Since
$w$ and $cu_1$ are both irreducible, the first redex in $wcu_1$
must cross the boundary between them. Hence we may write
$w=w_1w_2$ and $u_1=u_{11}u_{12}$, with $w_2\in\Gamma^+$, such that
\[
w_2cu_{11}=xy
\]
for some rule $xy\to y$ in $R$.

The letter $c$ must occur in the part $x$. Indeed, if $c$ occurred
in $y$, then $w_2=xs$ and $y=scu_{11}$ for some $s\in\Gamma^*$.
The first rewrite step would then give
\[
w_1xscu_1\to_S w_1scu_1,
\]
so $w_1s$ would be a shorter witness than $w$, a contradiction.

Therefore $x=w_2cz$ for some $z\in\Gamma^*$. From
$w_2cu_{11}=xy$, it follows that $u_{11}=zy$. Thus
$u_1=zyu_{12}$, and the required form is obtained by taking
$t=u_{12}$.

Conversely, suppose that a rule $xy\to y$ and words $w_2,z,t$
satisfy $x=w_2cz$ and $u_1=zyt$. Then
\[
(zw_2)cu_1
=zw_2czyt
=zxy t
\to_S zyt
=u_1.
\]
Hence $cu_1\preceq_S^r u_1$.

Finally, this characterization yields a decision procedure. For
each rule $xy\to y$ in the finite set $R$, enumerate the finitely
many decompositions $x=w_2cz$ with $w_2\in\Gamma^+$, and check
whether $zy$ is a prefix of $u_1$. Therefore
$cu_1\preceq_S^r u_1$ is decidable.
\end{proof}

\subsection{Proof of Lemma~\ref{lem:recursiveminsufmul}}

\begin{proof}
Let $u=cu_1$. If $u\not\preceq_S^r u_1$, then no proper suffix
$v$ of $u$ can belong to $\operatorname{SufMul}_S(u)$. Indeed, since
$v$ is a suffix of $u_1$, we may write $u_1=pv$. If
$u\preceq_S^r v$, then $qu\equiv_S v$ for some $q\in\Gamma^*$, and
hence $pqu\equiv_S u_1$, contradicting
$u\not\preceq_S^r u_1$. Therefore $u_{\min}=u$.

Suppose instead that $u\preceq_S^r u_1$. Since $u=cu_1$, we also
have $u_1\preceq_S^r u$. Thus, for every suffix $v$ of $u_1$,
\[
u\preceq_S^r v
\quad\Longleftrightarrow\quad
u_1\preceq_S^r v.
\]
Moreover, $u_1\in\operatorname{SufMul}_S(u)$ and $|u_1|<|u|$.
Hence the shortest element of $\operatorname{SufMul}_S(u)$ is
$(u_1)_{\min}$.

The condition $cu_1\preceq_S^r u_1$ is decidable by the preceding
lemma. Therefore the displayed recursion is effective, and it
terminates because each recursive call removes the first letter of
the input word.
\end{proof}

\subsection{Complexity of the Right-Divisibility Algorithm for Prefix-Erasing STS}
\label{app:rdiv-complexity}

We give a more explicit complexity analysis of the procedure used in the proof of Theorem~\ref{thm:rdiv-prefix-erasing}. Let $S=(\Gamma,R)$, let $\rho=|R|$, let $L=\max\{|\ell|:\ell\to r\in R\}$, and let $m=|u|+|v|$.

The procedure has three stages:
\begin{enumerate}
\item compute the normal forms $u^{\downarrow}$ and $v^{\downarrow}$;
\item compute $(u^{\downarrow})_{\min}=\operatorname{MinSufMul}_S(u^{\downarrow})$;
\item test whether $(u^{\downarrow})*{\min}$ is a suffix of $v^{\downarrow}$.
\end{enumerate}

\paragraph*{Normalization.} Since every prefix-erasing system is length-reducing, normal forms under a fixed finite convergent system can be computed in linear time. Thus, $u^{\downarrow}$ and $v^{\downarrow}$ can be computed in time $O_S(|u|+|v|)$. Writing $\rho=|R|$, $L=\max\{|\ell|:\ell\to r\in R\}$, and $m=|u|+|v|$, the dependence on $S$ can be made explicit as $O(\rho L+Lm)$, including preprocessing of the rule matcher.

\paragraph*{Minimal suffix multiplier.}
For each suffix $cu_1$ of $u^{\downarrow}$, the algorithm tests whether $cu_1\preceq_S^r u_1$. By the preceding characterization, this amounts to checking whether there exist a rule $xy\to y$ and a factorization $x=w_2cz$, with $w_2\in\Gamma^+$, such that $zy$ is a prefix of $u_1$.

A direct implementation takes $O(\rho L^2)$ time per suffix and hence $O(\rho L^2|u^{\downarrow}|)$ time in total. For fixed $S$, the relevant factorizations can be precomputed, so each test takes constant time and $\operatorname{MinSufMul}_S(u^{\downarrow})$ is computed in $O(|u^{\downarrow}|)$ time. Here suffixes are represented by positions in the word, avoiding any copying overhead.

\paragraph*{Final suffix test.}
Testing whether $(u^{\downarrow})_{\min}$ is a suffix of $v^{\downarrow}$ takes $O(|u^{\downarrow}|+|v^{\downarrow}|)$ time.

\paragraph*{Combined bound.}
The complete procedure therefore runs in time
$O(\rho L+Lm+\rho L^2|u^{\downarrow}|+|u^{\downarrow}|+|v^{\downarrow}|)=O(\rho L^2(m+1))$.
In particular, for fixed $S$, it runs in time $O(|u|+|v|)$. If the normal forms are supplied, the remaining computation is linear in $|u^{\downarrow}|+|v^{\downarrow}|$.

\subsection{Proof of Lemma~\ref{lem:backward-delta}}

\begin{proof}
    We prove the claim by backward induction on $k$.

    For $k=n+1$, we have $B_{n+1}^S(u,v)=\{v\}$. On the other hand,
    \[
        \{x\in \operatorname{IRR}(S) : x\varepsilon\equiv_S v\}
        =
        \{x\in \operatorname{IRR}(S) : x\equiv_S v\}.
    \]
    Since $S$ is convergent and both $x$ and $v$ are irreducible, $x\equiv_S v$ holds if and only if $x=v$. Hence the claim holds for $k=n+1$.

    Now suppose the claim holds for $k+1$. We prove it for $k$. Let $x\in B_k^S(u,v)$. By definition of $B_k^S(u,v)$, there exists $z\in B_{k+1}^S(u,v)$ such that $x\in \delta(u_k,z)$. By definition of $\delta$, this means $xu_k\equiv_S z$. By the induction hypothesis applied to $z\in B_{k+1}^S(u,v)$, we have $zu_{k+1}\cdots u_n\equiv_S v$. Therefore
    \[
        xu_ku_{k+1}\cdots u_n
        \equiv_S
        zu_{k+1}\cdots u_n
        \equiv_S
        v.
    \]
    Hence $x$ belongs to the right-hand side.

    Conversely, suppose $x\in \operatorname{IRR}(S)$ and $xu_ku_{k+1}\cdots u_n\equiv_S v$. Let $z=(xu_k)^\downarrow$ be the normal form of $xu_k$. Then $z\in \operatorname{IRR}(S)$ and $xu_k\equiv_S z$. Hence
    \[
        zu_{k+1}\cdots u_n
        \equiv_S
        xu_ku_{k+1}\cdots u_n
        \equiv_S
        v.
    \]
    By the induction hypothesis, $z\in B_{k+1}$. Since $xu_k\equiv_S z$, we have $x\in \delta(u_k,z)$. Therefore $x\in B_k$.

    This proves the claim for every $k=1,2,\dots,n+1$.
\end{proof}

\subsection{Proof of Corollary~\ref{cor:rightdivisibility-b1}}

\begin{proof}
    By definition, $u_ku_{k+1}\cdots u_n \preceq^r_S v$ means that there exists $w\in \Gamma^*$ such that $wu_ku_{k+1}\cdots u_n\equiv_S v$.

    If such a word $w$ exists, let $x=w^\downarrow$. Then $x\in \operatorname{IRR}(S)$ and $xu_ku_{k+1}\cdots u_n\equiv_S v$. By the proposition, $x\in B_k^S(u,v)$, so $B_k^S(u,v)\neq \varnothing$.

    Conversely, if $B_k^S(u,v)\neq \varnothing$, choose $x\in B_k$. By the proposition, $xu_ku_{k+1}\cdots u_n\equiv_S v$. Hence $u_ku_{k+1}\cdots u_n \preceq^r_S v$.

    Taking $k=1$ gives $u\preceq^r_S v$ if and only if $B_1\neq \varnothing$.
\end{proof}

\subsection{Proof of Lemma~\ref{lem:finite-delta}}

\begin{proof}
Fix $c \in \Gamma$ and $u \in \operatorname{IRR}(S)$. Recall that $\delta(c,u)=\{v \in \operatorname{IRR}(S) : vc \equiv_S u\}$. We show that there are only finitely many possible words $v$ in this set.

Let $v \in \delta(c,u)$. Then $v$ is irreducible and $vc \equiv_S u$. Since $v$ is irreducible, the word $v$ contains no redex. Hence any redex in $vc$, if one exists, must involve the final letter $c$. Therefore there are two cases.

First, suppose that $vc$ is irreducible. Since $u$ is also irreducible and $vc \equiv_S u$, convergence of $S$ implies that $vc=u$. Thus, in this case, $v$ is uniquely determined: it is obtained from $u$ by deleting the final letter $c$, if $u$ ends in $c$. If $u$ does not end in $c$, then there is no such $v$ in this case. Hence this case contributes at most one possible value of $v$.

Now suppose that $vc$ is reducible. Since $v$ is irreducible, every redex in $vc$ must involve the final letter $c$. Choose one such redex. Since this redex contains the final letter of $vc$, it occurs as a suffix of $vc$. Hence there exist a word $w$ and a rule $xy \to x$ in $R$, with $y \neq \epsilon$, such that $vc=wxy$. Since this occurrence of $xy$ ends with the final letter $c$, we may write $y=y'c$ for some word $y' \in \Gamma^*$. Hence $v=wxy'$.

Applying the rule $xy \to x$ gives $vc=wxy \to wx$. Since $v=wxy'$, the word $wx$ is a prefix of $v$. Because $v$ is irreducible, every factor of $v$ is irreducible, so $wx$ is irreducible. Also, $wx \equiv_S vc \equiv_S u$. Since both $wx$ and $u$ are irreducible, convergence of $S$ implies that $wx=u$.

Therefore, in the reducible case, every possible $v$ has the form $v=wxy'=uy'$, where $xy \to x$ is a rule in $R$ and $y=y'c$. Once the rule $xy \to x$ is fixed, the word $y'$ is fixed, and $v$ can occur only if $u$ ends in $x$, in which case $w$ is uniquely determined by $u=wx$. Hence each rule of $R$ contributes at most one possible value of $v$.

Since $R$ is finite, the reducible case contributes only finitely many possible values of $v$. Together with the at most one value from the irreducible case, this proves that $\delta(c,u)$ is finite.
\end{proof}

\subsection{Optimizing The Forward Search}\label{app:optimizing}

\begin{proposition}[Direct jump]
Let $u=cu_1\in \operatorname{IRR}(S)$, where $c\in \Gamma$ and $u_1\in \Gamma^*$. Suppose there exist a rule $xy\to y$ in $R$, words $w_2\in \Gamma^+$ and $z,t\in \Gamma^*$ such that $x=w_2cz$ and $u_1=zyt$. Then $u\preceq_S^r yt$.

In particular, $yt$ is a suffix of $u$ and is right-divisible from $u$.
\end{proposition}

\begin{proof}
Since $u=cu_1=czyt$ and $x=w_2cz$, we have $w_2u=w_2czyt=xyt$. Hence $w_2u=xyt\to_S yt$. Therefore $u\preceq_S^r yt$.
\end{proof}

\begin{proposition}[Jump correctness]
Let $u\in \operatorname{IRR}(S)$. Suppose $s$ is a suffix of $u$ such that $u\preceq_S^r s$. Then $u_{\operatorname{min}}=s_{\operatorname{min}}$.

In particular, in the situation of the previous proposition, if $s=yt$, then $u_{\operatorname{min}}=s_{\operatorname{min}}$.
\end{proposition}

\begin{proof}
Since $s$ is a suffix of $u$, there exists $p\in \Gamma^*$ such that $u=ps$. Hence $s\preceq_S^r u$.

By assumption, we also have $u\preceq_S^r s$. Therefore $u$ and $s$ are mutually right-divisible. Hence, for every suffix $v$ of $s$, we have $u\preceq_S^r v$ if and only if $s\preceq_S^r v$. Thus the suffixes of $s$ that lie in $\operatorname{SufMul}_S(u)$ are exactly the suffixes of $s$ that lie in $\operatorname{SufMul}_S(s)$.

Moreover, since $s\in \operatorname{SufMul}_S(u)$, the shortest element of $\operatorname{SufMul}_S(u)$ must occur among the suffixes of $s$. Therefore $u_{\operatorname{min}}=s_{\operatorname{min}}$.
\end{proof}

\begin{proposition}[Optimised recursive step]
Let $u=cu_1\in \operatorname{IRR}(S)$. Define $J(u)$ to be the set of all words $yt$ such that there exist a rule $xy\to y$ in $R$, words $w_2\in \Gamma^+$ and $z,t\in \Gamma^*$ satisfying $x=w_2cz$ and $u_1=zyt$.

If $J(u)=\varnothing$, then $u_{\operatorname{min}}=u$. If $J(u)\neq \varnothing$, then for every $s\in J(u)$, we have $u_{\operatorname{min}}=s_{\operatorname{min}}$.
\end{proposition}

\begin{proof}
By the finite test, $J(u)=\varnothing$ is equivalent to $cu_1\not\preceq_S^r u_1$. Hence, by the previous recursive characterisation, $u_{\operatorname{min}}=u$.

Now suppose $J(u)\neq \varnothing$, and let $s\in J(u)$. By the direct jump proposition, $u\preceq_S^r s$. Also, $s$ is a suffix of $u$. Hence, by the jump correctness proposition, $u_{\operatorname{min}}=s_{\operatorname{min}}$.
\end{proof}

\begin{algorithm}[H]
\caption{Compute $u_{\operatorname{min}}$ with jumps}
\begin{algorithmic}[1]
\Require An irreducible word $u\in \Gamma^*$
\Ensure The shortest word $u_{\operatorname{min}}$ in $\operatorname{SufMul}_S(u)$

\State $m\gets u$

\While{$m\neq \varepsilon$}
    \State Write $m=cm_1$, where $c\in \Gamma$ and $m_1\in \Gamma^*$
    \State $J\gets \varnothing$

    \ForAll{rules $xy\to y$ in $R$}
        \ForAll{factorisations $x=w_2cz$ with $w_2\in \Gamma^+$ and $z\in \Gamma^*$}
            \If{$zy$ is a prefix of $m_1$}
                \State Write $m_1=zyt$
                \State $J\gets J\cup\{yt\}$
            \EndIf
        \EndFor
    \EndFor

    \If{$J=\varnothing$}
        \State \Return $m$
    \Else
        \State Choose $s\in J$ of smallest length
        \State $m\gets s$
    \EndIf
\EndWhile

\State \Return $\varepsilon$
\end{algorithmic}
\end{algorithm}

\begin{proposition}[Complexity of the jump algorithm]
Let $S=(\Gamma,R)$ be a finite convergent prefix-erasing STS. Let $\rho=|R|$ and let $L=\max\{|l|:(l,r)\in R\}$. Given an irreducible input word $u\in \Gamma^*$ of length $n$, the jump algorithm computes $u_{\operatorname{min}}$ in time $O(\rho L^2 n)$ under a straightforward implementation.

In particular, if $S$ is fixed, then the algorithm runs in time $O(n)$.
\end{proposition}

\begin{proof}
At each iteration, the algorithm considers the current word $m$. If $m=\varepsilon$, it terminates. Otherwise, write $m=cm_1$, where $c\in \Gamma$ and $m_1\in \Gamma^*$.

The algorithm then searches through all rules $xy\to y$ in $R$. There are $\rho$ such rules. For each rule, it enumerates all factorisations $x=w_2cz$ with $w_2\in \Gamma^+$ and $z\in \Gamma^*$. Since $x$ is a subword of the left-hand side $xy$, we have $|x|\leq L$. Hence there are at most $L$ possible factorisations of this form.

For each such factorisation, the algorithm checks whether $zy$ is a prefix of $m_1$. Since $zy$ is a suffix of the left-hand side $xy$, we have $|zy|\leq L$. Therefore each prefix check costs $O(L)$ time. Thus one full search for possible jumps costs $O(\rho L^2)$.

If no jump is found, the algorithm terminates. If a jump is found, the algorithm replaces $m$ by some $s\in J$, where $s=yt$ and $m=czyt$. Since $c\in \Gamma$, the new word $s$ is a proper suffix of $m$. Hence $|s|<|m|$. Therefore each successful iteration strictly decreases the length of the current word by at least $1$.

Since the initial word has length $n$, there can be at most $n$ successful iterations. Therefore the total cost of all iterations is $O(\rho L^2 n)$. If $S$ is fixed, then $\rho$ and $L$ are constants, so the running time is $O(n)$.
\end{proof}

\begin{corollary}
Let $S=(\Gamma,R)$ be a finite convergent prefix-erasing STS, and let $u,v\in \Gamma^*$ be input words with $m=|u|+|v|$. The right-divisibility algorithm using the jump computation of $u_{\operatorname{min}}$ runs in time $O(\rho Lm^2+\rho L^2m)$ under a straightforward implementation. If $S$ is fixed, this is $O(m^2)$.
\end{corollary}

\begin{proof}
Computing normal forms of $u$ and $v$ costs $O(\rho Lm^2)$ under the straightforward rewriting implementation. By the proposition above, computing $\widehat{u}_{\operatorname{min}}$ from the irreducible word $\widehat{u}$ costs $O(\rho L^2m)$. Finally, checking whether $\widehat{u}_{\operatorname{min}}$ is a suffix of $\widehat{v}$ costs $O(m)$. Hence the total running time is $O(\rho Lm^2+\rho L^2m)$.
\end{proof}

\begin{proposition}[Complexity on irreducible inputs]
Let $S=(\Gamma,R)$ be a finite convergent prefix-erasing STS. Let $\rho=|R|$ and let $L=\max\{|l|:(l,r)\in R\}$. Suppose the input words $u,v\in \Gamma^*$ are already irreducible. Then the right-divisibility algorithm using the jump computation of $u_{\operatorname{min}}$ runs in time $O(\rho L^2|u|+|v|)$ under a straightforward implementation.

In particular, if $S$ is fixed, then the algorithm runs in time $O(|u|+|v|)$.
\end{proposition}

\begin{proof}
Since $u$ and $v$ are already irreducible, we do not need to compute their normal forms. The algorithm only needs to compute $u_{\operatorname{min}}$ and then check whether $u_{\operatorname{min}}$ is a suffix of $v$.

By the complexity analysis of the jump algorithm, computing $u_{\operatorname{min}}$ from an irreducible word $u$ costs $O(\rho L^2|u|)$. After that, checking whether $u_{\operatorname{min}}$ is a suffix of $v$ costs $O(|v|)$.

Therefore the total running time is $O(\rho L^2|u|+|v|)$. If $S$ is fixed, then $\rho$ and $L$ are constants, so this becomes $O(|u|+|v|)$.
\end{proof}

\subsection{Complexity of backward delta algorithm}
\label{app:backwardcomplexity}

\begin{proposition}[Complexity of the backward algorithm]
Let $S=(\Gamma,R)$ be a finite convergent suffix-erasing STS. Let $u,v \in \Gamma^*$, and write $u=u_1u_2\cdots u_n$. Let $m=|R|$ and let $M=\max\{|\ell| : (\ell,r)\in R\}$.

The backward algorithm deciding whether $u \preceq_S^r v$ runs in time exponential in $|u|$. More precisely, if $b \leq m+1$ is an upper bound on the number of possible elements of each set $\delta(c,z)$, then the number of words generated by the algorithm is at most $1+b+b^2+\cdots+b^n$, and hence at most $O(b^n)$. Therefore the algorithm runs in time $O(b^n \cdot p(|u|,|v|,|S|))$ for some polynomial $p$.

In particular, for fixed $S$, the algorithm runs in time $O(C^{|u|}\cdot p(|u|,|v|))$ for some constant $C$ depending only on $S$.
\end{proposition}

\begin{proof}
We first replace $u$ and $v$ by their normal forms. Since $S$ is finite and length-reducing, normal forms are computable effectively.

Now write $u=u_1u_2\cdots u_n$. The algorithm starts with $B_{n+1}=\{v\}$ and computes $B_n,B_{n-1},\dots,B_1$ by setting $B_k=\bigcup_{z\in B_{k+1}}\delta(u_k,z)$.

For suffix-erasing systems, each set $\delta(c,z)$ is finite. More precisely, if $z$ is irreducible, then any $x\in\delta(c,z)$ arises in one of the following ways. First, if $z$ ends in $c$, then $x$ may be obtained by deleting the final $c$ from $z$. Second, for each rule of the form $pq\to p$, where $q=q'c$, there is at most one possible predecessor $x$, namely the word $zq'$, provided that it satisfies the required suffix condition and is irreducible. Thus each rule contributes at most one possible predecessor, and the case where no rewrite occurs contributes at most one more predecessor. Hence $|\delta(c,z)|\leq |R|+1=m+1$. Let $b=m+1$.

It follows that $|B_k|\leq b|B_{k+1}|$ for every $k$. Since $|B_{n+1}|=1$, we obtain $|B_k|\leq b^{n+1-k}$. Therefore the total number of words generated while computing all the sets $B_{n+1},B_n,\dots,B_1$ is at most $1+b+b^2+\cdots+b^n$, which is $O(b^n)$.

For each generated word $z$, the set $\delta(c,z)$ can be computed by checking finitely many candidates, one for each rule of $S$ together with the possible no-rewrite case. These checks involve only suffix tests, irreducibility tests, and comparisons of words whose lengths are bounded by $|v|+O(nM)$. Hence the cost of processing each generated word is polynomial in $|u|$, $|v|$, and the size of $S$.

Therefore the total running time is $O(b^n\cdot p(|u|,|v|,|S|))$ for some polynomial $p$. Since $n=|u|$, this is exponential in $|u|$. If $S$ is fixed, then $b$ and $M$ are constants, so the running time is exponential in $|u|$ with a constant base depending only on $S$.
\end{proof}

\subsection{Proof of Proposition~\ref{prop:notlocallystable}}

\begin{proof}
    Consider $R=\{f(f(x,y),z) \to f(x,z)\}$. This is a variable lifting theory since we can take $C[\circ]=f(\circ, z)$, $D[\circ] = f(\circ,y)$ and we have that $l=f(f(x,y),z)=C[D[x]]$ and $r=f(x,z)=C[x]$. 

    Consider frame $\varphi = \nu a,b.\{x_1 \mapsto f(a,b)\}$. Here, $c_R=|f(f(x,y),z)|=5$. Suppose there exists a set $\operatorname{sat}(\varphi)$ satisfying (i), (ii), (iii), and (iv).

    First, $f(a,b) \in \operatorname{sat}(\varphi)$ by (i).

    Consider infinitely many terms $t_0,t_1,\dots$ defined by recurrence $t_0=f(a,b)$, and $t_{k+1}=f(a,t_k)$ for $k \ge 1$.
     We prove that $t_k \in \operatorname{sat}(\phi)$ for all $k \ge 0$.
     Suppose $t_k \in \operatorname{sat}(\phi)$. Note that $f(f(a,b),t_k) \xrightarrow{\varepsilon}_R f(a,t_k)=t_{k+1}$. By (iii), we have $t_{k+1} \xrightarrow{*}_R C'[S_1',\dots,S_l']$ for some context $C'$ with $|C'| \le 25$ and $fn(C') \cap \{a,b\} =\emptyset$ and $S_1',\dots,S_l' \in \operatorname{sat}(\varphi)$. Since $t_k$ is in normal form, then $t_k=C'[S_1',\dots,S_l']$. If $C'$ is not a hole, then $C'=f(C_1',C_2')$ for some context $C_1'$ and $C_2'$. Then $C_1'=a$, but hence $a \in fn(C')$, a contradiction with $fn(C') \cap \{a,b\} = \emptyset$. Hence $C'$ is a hole, and thus $l=1$ and $S_1'=f(a,f(a,b)) \in \operatorname{sat}(\varphi)$.

     This means that the set $\operatorname{sat}(\varphi)$ must contain infinitely many terms, which is a contradiction. 
\end{proof}

\subsection{Proof of Proposition~\ref{prop:homemb}}

\begin{proof}
Let $l\to r$ be an SSL rule, so that $l=C[D_1[t_1],\ldots,D_n[t_n]]$ and $r=C[t_1,\ldots,t_n]$. For each $i$, repeatedly applying the third clause in the definition of homeomorphic embedding gives $D_i[t_i]\succeq_{\mathit{emb}}t_i$. Applying the second clause along the structure of the common outer context $C$ then yields $l\succeq_{\mathit{emb}}r$. Hence every rule of the TRS is homeomorphic-embedded.

The remaining claims follow from the inclusions between the classes.
\end{proof}

\end{document}